\documentclass[12pt,reqno]{amsart}
\usepackage[dvips,colorlinks=true,linkcolor=red,citecolor=blue]{hyperref}
\usepackage[final]{graphicx}
\usepackage{amsfonts}
\usepackage{pdfsync}
\usepackage{amsmath}
\usepackage{amssymb}
\usepackage{amsthm}
\newtheorem{theorem}{Theorem}[section]
\newtheorem{lemma}[theorem]{Lemma}
\newtheorem{corollary}[theorem]{Corollary}
\newtheorem{proposition}[theorem]{Proposition}

\newcommand{\R}{{\mathord{\mathbb R}}}
\newcommand{\C}{{\mathord{\mathbb C}}}
\newcommand{\Z}{{\mathord{\mathbb Z}}}
\newcommand{\N}{{\mathord{\mathbb N}}}
\newcommand{\E}{{\mathord{\mathbb E}}}
\newcommand{\PP}{{\mathord{\mathbb P}}}
\newcommand{\supp}{{\mathop{\rm supp\ }}}
\newcommand{\tr}{{\mathop{\rm tr\ }}}
\newcommand{\re}{{\mathop{\rm Re }}}

\newcommand{\esp}{\mathbb{E}}

\begin{document}

\title[The random displacement model]{Localization for the random displacement model}

\author[Klopp]{Fr{\'e}d{\'e}ric Klopp$^1$}

\thanks{F.\ K.\ was partially supported by the ANR grant
  08-BLAN-0261-01.}

\author[Loss]{Michael Loss$^2$}

\thanks{M.\ L.\ was supported in part by NSF grant DMS-0901304.}

\author[Nakamura]{Shu Nakamura$^3$}

\thanks{S.\ N.\ was partially supported by JSPS grant Kiban (A) 21244008}

\author[Stolz]{G{\"u}nter Stolz$^4$}

\thanks{G.\ S.\ was supported in part by NSF grant DMS-0653374.}

\address{$^1$ LAGA, U.M.R. 7539 C.N.R.S, Institut Galil{\'e}e, Universit{\'e}
  de Paris-Nord, 99 Avenue J.-B.  Cl{\'e}ment, F-93430 Villetaneuse,
  France\ et \ Institut Universitaire de France}
\email{\href{mailto:klopp@math.univ-paris13.fr}{klopp@math.univ-paris13.fr}}

\address{$^2$ Georgia Institute of Technology, School of Mathematics,
Atlanta, Georgia 30332-0160, loss@math.gatech.edu}
\email{\href{mailto:loss@math.gatech.edu}{loss@math.gatech.edu}}

\address{$^3$ Graduate School of Mathematical Sciences, University of Tokyo, 3-8-1 Komaba, Meguro-ku, Tokya, Japan 153-8914, shu@ms.u-tokyo.ac.jp}
\email{\href{mailto:shu@ms.u-tokyo.ac.jp}{shu@ms.u-tokyo.ac.jp}}

\address{$^4$ University of Alabama at Birmingham, Department of Mathematics, Birmingham, Alabama 35294-1170, stolz@uab.edu}
\email{\href{mailto:stolz@uab.edu}{stolz@uab.edu}}

% \date{\today}
\maketitle

\vspace{.3truein}
\centerline{\bf Abstract}

\medskip {\sl We prove spectral and dynamical localization for the
  multi-dimensional random displacement model near the bottom of its
  spectrum by showing that the approach through multiscale analysis is
  applicable. In particular, we show that a previously known Lifshitz
  tail bound can be extended to our setting and prove a new Wegner
  estimate. A key tool is given by a quantitative form of a property
  of a related single-site Neumann problem which can be described as
  ``bubbles tend to the corners''.}

\section{Introduction}
\label{sec:introduction}

We consider the {\it random displacement model} (RDM), a random Schr{\"o}dinger operator
\begin{equation} \label{eq:hamiltonian}
H_{\omega} = -\Delta + V_{\omega}
\end{equation}
in $L^2(\R^d)$, $d\ge 1$, where the random potential has the form
\begin{equation} \label{eq:potential}
V_{\omega}(x) = \sum_{i\in \Z^d} q(x-i-\omega_i).
\end{equation}

This models a random perturbation of the periodic potential $\sum_i q(x-i)$, where the single-site terms sit at exact lattices sites $i\in \Z^d$. The parameter $\omega = (\omega_i)_{i\in \Z^d}$ describes a configuration of random displacement vectors $\omega_i \in \R^d$. Before entering a more thorough discussion of background and assumptions, let us state the main result of our work in a simple non-trivial special case:

\begin{theorem}
Suppose that $d\ge 2$ and that $q\in C_0^{\infty}(\R^d)$ is real-valued, sign-definite, rotation symmetric and supp$\,q \subset \{x:|x|\le r\}$ for some $r<1/4$. Also assume that $(\omega_i)_{i\in \Z^d}$ are i.i.d.\ $\R^d$-valued random variables, uniformly distributed on $[-d_{max}, d_{max}]^d$ where $d_{max} = 1/2-r$.

Then $H_{\omega}$ is spectrally and dynamically localized at energies near the bottom of its almost sure spectrum.
\end{theorem}

The exact meaning of the latter will be recalled below.

The RDM represents a natural way to model a solid with structural disorder. It can be considered as intermediate between the {\it Anderson model}, which has no structural disorder and randomness instead appears in the form of coupling constants at the single site terms, and the {\it Poisson model}, where the structure of the medium is entirely dissolved by placing single-site scatterers at the points of a Poisson process. This point of view has recently been supported by an investigation of the integrated density of states of the RDM in \cite{F} and \cite{FU}.

It is physically expected that multi-dimensional random Schr{\"o}dinger operators such as the three models mentioned above should exhibit {\it localization} at energies near the bottom of the spectrum, while extended states should exist at high energy. For mathematicians the latter remains an open problem, while localization at low energy for the Anderson model (both in the continuum and in the original lattice setting) and, with more effort and more recently, for the Poisson model has been proven rigorously. The references given in Section~\ref{sec:msa} below may be used as a starting point into the enormous literature on localization for the Anderson model. For the multi-dimensional Poisson model localization at low energy has been proven, separately for the case of positive and negative single-site potentials, in \cite{GHK1} and \cite{GHK2}.

A localization proof for the RDM provides additional challenges and
does not follow from the methods alone which have led to proofs for
the Anderson and Poisson models. The main reason for this is that the
RDM does not have any obvious monotonicity properties with respect to
the random parameters. Such properties are frequently used in an
essential way in the theory of the Anderson model and, to some extend,
can also be exploited for the Poisson model.

This becomes apparent most immediately when attempting to characterize the bottom of the spectrum for these models. For the Anderson model with sign-definite potential $q$ this corresponds to choosing all couplings minimal (if $q$ is positive) or maximal (if $q$ is negative), respectively. For the Poisson model the bottom of the spectrum is $0$ if $q$ is positive (due to large regions devoid of any Poisson points) and $-\infty$ if $q$ is negative (due to dense clusters of Poisson points).

Identifying a mechanism which characterizes the bottom of the spectrum, a crucial preliminary step towards the localization question, poses a non-trivial challenge for the RDM. It is not at all obvious which configurations $\omega = (\omega_i)$ of the displacements should characterize minimal energy.

Far more than a characterization of the spectral minimum is needed for a proof of localization. From the theory of the Anderson model it is well known that sufficient ingredients are smallness of the integrated density of states (IDS) at the bottom of the spectrum (e.g.\ in the form of {\it Lifshitz tails}), and sensitivity of the spectrum to the random parameters (for example in the form of {\it spectral averaging} or {\it Wegner estimates}). The usual approaches to verifying both of these ingredients make heavy use of monotonicity properties as well, providing further obstacles to a localization proof for the RDM, where such properties are not apparent.

There are two previous works in which localization properties of modified versions of the RDM (\ref{eq:hamiltonian}), (\ref{eq:potential}) have been shown. In \cite{Klopp}, Klopp considered a semiclassical version $-h^2 \Delta + V_{\omega}$ of (\ref{eq:hamiltonian}) and was able to show the existence of a localized region in the spectrum for sufficiently small value of the semiclassical parameter $h$. In this regime neither Lifshitz tails nor an exact characterization of the spectral minimum are needed for the localization proof. More recently, Ghribi and Klopp \cite{GhKl} considered an RDM of the form (\ref{eq:hamiltonian}) with an additional periodic background potential. For a generic non-zero choice of the latter and sufficiently small displacement vectors $\omega_i$, they use first order perturbation arguments to recover monotonicity properties which lead to Lifshitz tails as well as a Wegner estimate, and thus localization.

Our goal here is to prove localization for the model (\ref{eq:hamiltonian}), (\ref{eq:potential}) without working in the semi-classical regime or modifying the background. This means that we can not hope to reveal any monotonicity properties by exclusively using perturbative arguments, at least not easily and not with first order perturbation theory.

We will make central use of symmetry properties. In particular, we will assume throughout that the real-valued single-site potential $q$ is reflection symmetric in each coordinate, i.e.\
\begin{equation} \label{eq:symm}
q(\ldots, x_{k-1}, -x_k,\ldots, x_{k+1}, \ldots) = q(\ldots, x_{k-1}, x_k, x_{k+1}, \ldots), \quad k=1,\ldots,d.
\end{equation}

Due to the use of some higher order perturbation theory we will need some smoothness of $q$ and for convenience assume that
\begin{equation} \label{eq:qsmooth}
q\in C^{\infty}(\R^d),
\end{equation}
which could be substantially weakened. We also assume that
\begin{equation} \label{eq:qsupp}
\mbox{supp}\,q \subset [-r,r]^d \quad \mbox{for some $r<1/4$}.
\end{equation}

The latter is best understood in conjunction with the following
assumption on the displacement parameters $\omega = (\omega_i)_{i\in
  \Z^d}$: They are i.i.d.\ $\R^d$-valued random variables, distributed
according to a measure $\mu$ satisfying
\begin{equation}
  \label{eq:rhosupp}
{\mathcal C}  \subset \mbox{supp}\,\mu \subset [-d_{max},d_{max}]^d,
\end{equation}
where $d_{max} = \frac{1}{2}-r$. Here ${\mathcal C} := \{(\pm d_{max}, \ldots, \pm d_{max})\}$ denotes the $2^d$ corners of the cube $[-d_{max}, d_{max}]^d$. For the proof of our main result, Theorem~\ref{thm:main} below, we will need some smoothness of the distribution of $\mu$. A convenient assumption is that $\mu$ has a density $\rho$, i.e. $\mu(B) = \int_B \rho(\omega_0)\,d\omega_0$ for all Borel subsets $B$ of $[-d_{max},d_{max}]^d$, and that
\begin{equation} \label{eq:rhosmooth}
\rho \in C^1([-d_{max},d_{max}]^d).
\end{equation}
However, as will be discussed in Section~\ref{sec:conclusion}, our proofs work for considerably more general distributions. In particular, our proof of Theorem~\ref{thm:main} below only uses that $\mu$ has a $C^1$-density in a neighborhood of the corners ${\mathcal C}$ and can be arbitrary away from this neighborhood.

By (\ref{eq:qsupp}) and (\ref{eq:rhosupp}), the centers of the
``bubbles'' $q(x-i-\omega_i)$ can move all the way into the
corners $i + {\mathcal C}$ of
$i+[-d_{max},d_{max}]^d$, and the supports of the bumps stay, up
to touching, mutually disjoint. See Figure~\ref{fig1} for a typical
configuration. Note that only $r<1/2$ would be needed to do this in a
non-trivial way (give the bubble space to move), but that we will need
$r<1/4$ for technical reasons to make use of results in \cite{BLS2} and \cite{KN2}, see the discussion preceding Corollary~\ref{le:3} below.

\vspace{.3cm}

\begin{figure}[h]
  \centering
  \includegraphics[width=5cm]{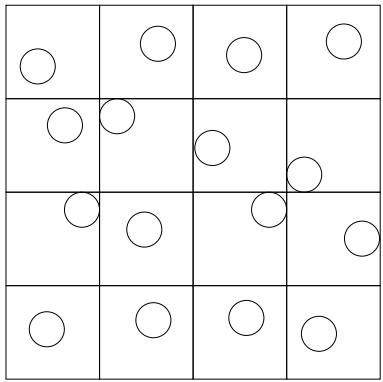}
  \caption{A typical configuration}
  \label{fig1}
\end{figure}

\vspace{.3cm}

The random operator $H_{\omega}$ is ergodic with respect to shifts in $\Z^d$ and thus, by the general theory of ergodic operators (see e.g.~\cite{Stollmann}), its spectrum is almost surely deterministic, i.e.\ there is a $\Sigma \subset \R$ such that
\[ \sigma(H_{\omega}) = \Sigma \quad \mbox{almost surely}.\]

Under the above assumptions it was shown in \cite{BLS1} that among all configurations $\omega$ a configuration with lowest spectral minimum is given by $\omega^* = (\omega^*_i)_{i\in \Z^d}$ where
\begin{equation} \label{eq:omegastar}
\omega^*_i = ((-1)^{i_1} d_{max}, \ldots, (-1)^{i_d} d_{max}), \quad i=(i_1,\ldots,i_d) \in \Z^d.
\end{equation}

As discussed in \cite{BLS1}, under the condition (\ref{eq:rhosupp}) this also characterizes the minimum of the almost sure spectrum of the RDM, i.e.\
\[ E_0 := \inf \Sigma = \inf \sigma(H_{\omega^*}).\]
In the configuration $\omega^*$ the single site potentials in $V_{\omega^*}$ form densest possible clusters where $2^d$ neighboring bumps move into adjacent corners of their unit cells, see Figure~\ref{fig2}.

\vspace{.3cm}

\begin{figure}[h]
  \centering
  \includegraphics[width=5cm]{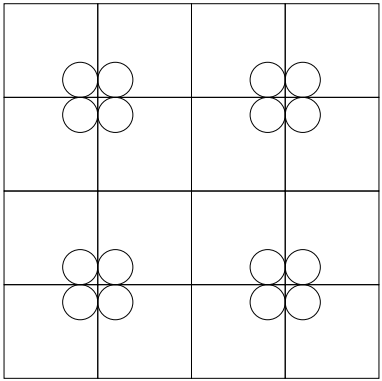}
  \caption{The minimizing configuration}
  \label{fig2}
\end{figure}

\vspace{.3cm}

Thus, for the set of assumptions listed above, we have an answer to the preliminary question of characterizing $E_0 = \inf \Sigma$. For a proof of localization we will need much more information.  In particular, we will need quantitative bounds on the probability that other configurations have spectral minimum close to $E_0$. This will require one more condition on $q$.

To state this condition, we need to introduce Neumann operators where a single bump is placed into a unit cell at varying position. For this, and for later, set
\[ \Lambda_r = \Lambda_r(0) = \left(-\frac{r}{2}, \frac{r}{2}\right)^d, \quad \Lambda_r(x) = \Lambda_r+x, \quad \chi_x = \chi_{\Lambda_1(x)},\]
the characteristic function of the unit cube centered at $x$.

For $a\in [-d_{max},d_{max}]^d$ let
\begin{equation} \label{eq:singlesiteop}
H_{\Lambda_1}^N(a) = -\Delta + q(x-a)  \quad \mbox{on $L^2(\Lambda_1)$}
\end{equation}
 with Neumann boundary condition. Finally, let
\[ E_0(a) = \inf \sigma(H_{\Lambda_1}^N(a)) \]
be the lowest eigenvalue of $H_{\Lambda_1}^N(a)$. Note that, by symmetry of $q$, $E_0(a)$ is symmetric with respect to all the coordinate hyperplanes $a_k=0$, $k=1,\ldots,d$.

In \cite{BLS1} the following alternative was established (only requiring boundedness of $q$, not smoothness):
\begin{itemize}
\item[(i)] Either $E_0(a)$ is strictly maximized at $a=0$ and strictly minimized in the corners ${\mathcal C}$ of $[-d_{max},d_{max}]^d$,
\item [(ii)] or $E_0(a)$ is identically zero, and the ground state of $H_{\Lambda_1}^N(a)$ is constant near the boundary of $\Lambda_1$.
\end{itemize}

From this and the symmetries of $q$ it is easy to see, e.g.\ \cite{BLS1}, that $\omega^*$ given by (\ref{eq:omegastar}) is a spectrally minimizing configuration and that the almost sure spectral minimum $E_0$ of $H_{\omega}$ is given by the minimum value of $E_0(a)$, i.e.\ its value at $a\in {\mathcal C}$.

While it is possible to construct non-vanishing $q$ where (ii) holds, e.g.\ our remarks in Section~\ref{sec:discussion}, alternative (i) is the generic case. It holds if $q\not= 0$ is sign-definite (since then the ground state energy $0$ of the Neumann Laplacian $-\Delta_{\Lambda}^N$ must be shifted up or down), but also for generic sign-indefinite $q$.

We are now able to state our main result on localization. Here, for
a self-adjoint  operator $H$ and Borel function $g$ we define $g(H)$
by the functional calculus. In particular, $\chi_I(H)$ denotes the
spectral projection onto $I$ for $H$.

\begin{theorem} \label{thm:main} Assume that $d\ge 2$, $\omega$ and
  $q$ satisfy (\ref{eq:symm}) to (\ref{eq:rhosmooth}), and that
  $E_0(a)$ does not vanish identically in $a\in [-d_{max},d_{max}]^d$.

  Then there exists $\delta>0$ such that $H_{\omega}$ almost surely
  has pure point spectrum in $I=[E_0, E_0+\delta]$ with exponentially
  decaying eigenfunctions. Moreover, $H_{\omega}$ is dynamically
  localized in $I$, in the sense that for every $\zeta <1$, there
  exist $C<\infty$ such that
  \begin{equation}
    \label{eq:dynloc}
    \E \left( \sup_{|g|\le 1} \|\chi_x g(H_{\omega}) \chi_I(H_{\omega})
      \chi_y\|_2^2 \right) \le C e^{-|x-y|^{\zeta}}
  \end{equation}
  for all $x,y\in \Z^d$. Here, the supremum is taken over all Borel
  functions $g:\R\to \C$ which satisfy $|g|\le 1$ pointwise.
\end{theorem}

Note here that dynamical localization in physical sense is a special case of (\ref{eq:dynloc}), choosing $g(H) = e^{-itH}$ and taking the supremum over $t\in \R$.

The proof of Theorem~\ref{thm:main} proceeds via multiscale analysis and our main task will be to establish the two main ingredients into the multiscale analysis, i.e.\ a smallness bound on the probability that finite volume restrictions of $H_{\omega}$ have low lying eigenvalues (related to Lifshitz tails of the integrated density of states) and a Wegner estimate.

It is in the proof of these two ingredients where new ideas are needed. In both proofs we will use that alternative (i) can be strengthened if some smoothness is assumed for $q$: in this case it can be shown that the first partial derivatives of $E_0(a)$ are non-zero in all directions away from the symmetry planes. In particular, at its strict minima in the corners ${\mathcal C}$, the function $E_0(a)$ is not flat as it has non-vanishing gradient. The proof of this result, given in Section~\ref{sec:key}, starts from a second-order perturbation theory formula.

This can be considered as the crucial monotonicity property which makes the localization proof work. In the spectrally minimizing configuration $\omega^*$ of the RDM all bubbles sit in corners. Non-vanishing of the gradient of $E_0(a)$ in the corners will allow to gain quantitative control on how close the ground state energies for other configurations are to $E_0$.

First, in Section~3, this will lead to a Lifshitz tail bound with a proof which is
based on an argument in \cite{KN2}. This also uses the fact that under
alternative (i) the configuration $\omega^*$ is, up to translation,
the unique {\it periodic} configuration with spectral minimum $E_0$, a
result established in \cite{BLS2}, where it is also shown that this is
{\it not} true for $d=1$. In \cite{KN2} it was shown that this
uniqueness result leads to a Lifshitz tail bound for the IDS near
$E_0$ if the distribution of the $\omega_i$ is discrete and contains
all corners ${\mathcal C}$. In fact, for
technical reasons the bound obtained there is weaker than the
``classical'' Lifshitz tail, but it is strong enough for a
localization proof (if one also has a Wegner estimate).  In
Section~\ref{sec:lifshitz} we will show how non-vanishing of $\nabla
E_0(a)$ in the corners can be used to get same bound on the IDS under
the assumptions considered here, i.e.\ with distribution satisfying
(\ref{eq:rhosupp}).

    What allows us to push the localization argument through is that under assumption (\ref{eq:rhosupp}) we can also prove a Wegner estimate, see Section~\ref{sec:wegner}. This will again use that $\nabla E_0(a) \not= 0$ at the corners, which will provide us with a measure for how much the ground state energy of finite volume restrictions of $H_{\omega}$ is pushed upwards if the bubbles $q(x-i-\omega_i)$ move away from the corners towards the center $i$ of $\Lambda_1(i)$. This is technically implemented in Proposition~\ref{prop:key} below in form of positivity of the derivative of $H_{\omega}$ with respect to a suitable vector field, whose proof is close to an argument previously used in \cite{KNNN} for operators with random magnetic fields. Based on this result, we are able to prove a Wegner estimate by modifying an argument developed in \cite{CHN} and \cite{HK} in the context of the Anderson model, which itself is a modification of the original argument due to Wegner. In particular, we obtain the correct (linear) volume dependence and can conclude H{\"o}lder continuity of the IDS as a by-product.

We finally mention that the connection between monotonicity properties of $H_{\Lambda_1}^N(a)$ and the RDM is made through surprisingly simple Neumann bracketing arguments, making crucial use of the variational characterization of the lowest eigenvalue as the minimum of the quadratic form. Thus, we are able to deduce monotonicity properties of a model with infinitely many parameters from a one-parameter model. This trick, employed in Sections~\ref{sec:lifshitz} and \ref{sec:wegner}, was previously used in \cite{KN1} to study the Anderson model with sign-indefinite single-site potential. It is the main reason why symmetry of the single-site potential is important for us.

Our concluding Section~\ref{sec:conclusion} serves two purposes. First, in Section~\ref{sec:msa}, we briefly discuss how the Lifshitz tail bound and Wegner estimate obtained here lead to a proof of spectral and dynamical localization via multi-scale analysis. In Section~\ref{sec:discussion} we mention some generalizations, related results and open problems.

\section{Bubbles tend to the corners} \label{sec:key}

In this section we will prove a property of the derivatives of $E_0(a)$, the ground state of the Neumann operators $H_{\Lambda_1}^N(a)$ defined in (\ref{eq:singlesiteop}), which is crucial for all our later arguments. The main result of \cite{BLS1}, i.e.\ that under the generic alternative (i) the function $E_0(a)$ is strictly minimized in the corners ${\mathcal C}$, may be dubbed as ``bubbles tend to the corners''. The seemingly small but important improvement to be shown here means that ``as bubbles tend to the corners the rate of change
of $E_0(a)$ does not vanish''.

In fact, in \cite{BLS1} two methods were developed to prove that the
minimium of $E_0(a)$ is in $\mathcal C$. One method, used to prove
Theorem~1.3 in \cite{BLS1}, applied to Neumann operators on
rectangular domains as considered here and is mostly based on
exploiting symmetries of the domain and the potential. A second,
very different, method was developed in \cite{BLS1} to show the
phenomenon that ``bubbles tend to the boundary'' for Neumann
operators on general smooth domains and with smooth potentials, see
the proof of Theorem~1.4 in \cite{BLS1}. What we do here is to apply
the second method in the setting of rectangular domains. That
``smooth methods'' apply to rectangles is possible due to our use of
Neumann boundary conditions, which allow to get smooth extensions of
the ground state eigenfunction by reflection. The specific geometry
of rectangles allows us to push the method used in Theorem~1.4 of
\cite{BLS1} further and to obtain the non-vanishing of derivatives
of $E_0(a)$ in the corners.

\subsection{Basic smoothness properties} \label{sec:basicsmooth}

Throughout this section we fix an open rectangular parallelepiped $D \subset \R^d$ and consider operators on $L^2(D)$. Let $q$ be a real-valued smooth function with closed support in $D$, i.e.\ $q\in C^{\infty}(D)$ such that $q(x)=0$ for $x$ in a neighborhood of $\partial D$.
Consider the quadratic form
$$
\Vert \nabla u \Vert_2^2 + \langle u , qu \rangle
$$
where we use the symbol
$$
\langle u,v \rangle = \int_D \overline{u(x)} v(x) dx \ .
$$
This quadratic form is a closed form on $H^1(D)$ and defines a
unique self--adjoint operator
$$
H :=-\Delta_N + q
$$
where $\Delta_N$ is the Neumann Laplacian. The eigenvalues of this
operator are discrete, have finite degeneracy and tend to infinity.
It will be important for us that the eigenfunctions are regular up to
and including the boundary, i.e.,  $C^\infty(\overline D)$. We state
this fact as a lemma.
\begin{lemma} \label{lem:efsmooth}
The eigenfunctions of the operator $H$ can be extended to a
neighborhood of $D$ where they are infinitely often differentiable.
\end{lemma}
\begin{proof}
By reflecting the potential and eigenfunctions repeatedly across the boundary one
obtains a generalized eigenfunction with the same eigenvalue on the whole space
$\R^d$. Since the potential is $C^\infty$ it follows, by elliptic
regularity, that the generalized eigenfunctions are in $C^\infty(\R^d)$.
\end{proof}

We set for $a \in \R^d$
$$
q_a(x)= q(x-a) \ .
$$
The set of points $a$ for which the support of $q_a$ is a subset of
$D$ is denoted by $G$. This is an open rectangle as well. Let
$$
H_a := -\Delta_N +q_a,
$$
which is short for the notation $H_D^N(a)$ used elsewhere in this
paper, and denote its eigenvalues in increasing order and counted
with multiplicity by
$$
E_n(a) \ , \ n=0,1,2, \dots
$$
and the corresponding real-valued normalized eigenfunctions as
$$
u_n(x,a) \ .
$$
In what follows we often denote the extension of the function to a
larger set by the same symbol.
\begin{lemma}[\bf{Differentiability of eigenvalue and eigenfunction}] \label{lem:diff}
The eigenvalue $E_0(a)$ as well as the eigenfunction $u_0(\cdot,a)$ are (as an $L^2(D)$-valued function)
infinitely often differentiable in a neighborhood of $G$. In
particular these two functions are in $C^\infty(\overline G)$.
\end{lemma}
\begin{proof} Continuity of the eigenvalues was proved in Lemma 2.1 in \cite{BLS1}. Since the ground
state eigenvalue is not degenerate, $E_0(a)$ is strictly separated from the rest of
the spectrum locally uniformly in $a$. Hence, there is a circle $C$
with center $E_0(a)$ which is strictly separated from the rest of
the spectrum, locally uniformly in $a$.
Using the formula, which holds for all $z$ in the intersection of the resolvent sets of $H_a$ and $H_{a'}$,
\begin{equation} \label{resolventformula}
(H_a -z)^{-1} -  (H_{a'} -z)^{-1} = (H_a -z)^{-1}\left[q_{a'} -q_a
\right] (H_{a'} -z)^{-1} \ ,
\end{equation}
we learn that the  projection onto the
eigenfunction $u_0(x,a)$,
$$
P_0 (a) = \frac{1}{2\pi i} \oint_C (H_a-z)^{-1} dz \ ,
$$
is also continuous.
Repeated use of (\ref{resolventformula}) shows that
the projection $P_0(a)$ can be arbitrarily often differentiated with respect to $a$.
Note that formula (\ref{resolventformula})  holds for all positions of the
potential, in particular, the support of the potential does not have
to be in $D$.   The
eigenvalue equation
$$
(H_a -z)^{-1} P_0(a) = (E_0(a) -z)^{-1} P_0(a)
$$
now shows that $E_0(a)$, in turn, can also be differentiated as often
as we please.
\end{proof}

Given the above smoothness properties and the non-degeneracy of $E_0(a)$ we can derive the following first and second order perturbation formulas for $E_0(a)$,
\begin{equation} \label{eq:fopt}
\partial_{a_j} E_0(a) = - \langle u_0, (\partial_{x_j} q_a) u_0 \rangle,
\end{equation}
and
\begin{equation} \label{eq:sopt}
\partial_{a_j}^2 E_0(a) = \langle u_0, (\partial_{x_j}^2 q_a) u_0\rangle - 2 \sum_{k\not= 0} \frac{\langle u_0, (\partial_{x_j} q_a) u_k\rangle^2}{E_k-E_0},
\end{equation}
for $j=1,\ldots,d$. This is done in complete analogy to the better known case of non-degenerate eigenvalues of operators of the form $A+\lambda B$, with (\ref{eq:fopt}) corresponding to the Feynman-Hellmann formula, see e.g.\ Section~2.3 of \cite{BLS1}.

\subsection{Second order perturbation theory}

Making notations more explicit, we now write
\[ D = (\alpha_1, \beta_1) \times D', \]
where
\[ D' = (\alpha_2, \beta_2) \times \ldots \times (\alpha_d, \beta_d) \subset \R^{d-1}.\]
Then $G = I \times G'$, with an open interval $I=(-\delta_1,\delta_2)$, $\delta_1>0$, $\delta_2>0$, and an open rectangular parallelepiped $G' \subset \R^{d-1}$.

We will consider the dependence of $E_0(a_1,\ldots,a_d)$ on the first variable $a_1 \in I$ at fixed value of $(a_2,\ldots,a_d) \in G'$. Throughout this subsection we will abuse notation and write $q_{a_1} := q_a$, $H_{a_1} := -\Delta_N + q_{a_1}$, as well as $E_n(a_1)$ and $u_n(x,a_1)$ for its eigenvalues and eigenfunctions.

For simplicity of presentation we focus on the variable $a_1$. Clearly, the arguments below also apply to the dependence of $E_0$ on each of the other variables $a_2$, \ldots, $a_d$.

The following lemma provides us with a differential equation for the
eigenvalue $E_0(a_1)$.

\begin{lemma} \label{perturbation}
The ground state energy satisfies the equation
\begin{equation} \label{eq:bigone}
E_0'' - 4\langle u_0, \partial_1 u_0\rangle E_0' =
 -  2\sum_{k \not= 0} \frac{B(u_k, \partial_1
u_0)^2}{E_k-E_0} \ .
\end{equation}
\end{lemma}

Here $B(u,v) := \langle u,\Delta v\rangle - \langle \Delta u,v\rangle$, $E_0'$ and $E_0''$ refer to $a_1$-derivatives of $E_0$, and $\partial_1 u_0$ refers to the spatial derivative in $x_1$-direction.

\begin{proof}
The proof is essentially the same as the proof of Lemma~5.1 in \cite{BLS1},  Due to the fact that we work with the partial derivative in a fixed coordinate direction instead of the gradient and that $D$ has flat faces, the argument here is simpler than in \cite{BLS1} by not having to include a term related to the curvature of the boundary of $D$.

Starting with (\ref{eq:sopt}), which is a ``partial'' version of (25) in \cite{BLS1}, one follows the argument there and arrives at a partial version of (27) in \cite{BLS1},
\begin{eqnarray} \label{eq:E0eq}
E_0'' &=&  2\sum_{k} B(u_k, \partial_1 u_0) \langle u_k, \partial_1 u_0\rangle
-
4\langle u_0, (\partial_1 q_{a_1}) u_0\rangle \langle u_0,\partial_1 u_0\rangle \nonumber\\
& & \mbox{}- 2\sum_{k \not= 0} \frac{B(u_k, \partial_1
u_0)^2}{E_k-E_0} \ .
\end{eqnarray}
As in \cite{BLS1} the first term term can be rewritten as
\begin{equation} \label{eq:defofB}
\sum_{k} B(u_k, \partial_1 u_0)  \langle u_k,\partial_1 u_0 \rangle =- \langle \partial_1 u_0, (-\Delta+ q_{a_1} -E_0)\partial_1u_0\rangle
+\sum_k (E_k-E_0) |\langle u_k,\partial_1u_0\rangle|^2 \ .
 \end{equation}
Since $\partial_1 u_0$ is in the form domain,
\begin{eqnarray*}
\sum_k (E_k-E_0) |\langle u_k,\partial_1u_0\rangle|^2 + E_0 \|\partial_1 u_0\|^2 & = & \Vert (H-E_0)^{1/2} \partial_1 u_0\Vert^2 + E_0 \|\partial_1 u_0\|^2 \\
& = & \int_D \left[ |\nabla \partial_1 u_0|^2 + q_{a_1} (\partial_1u_0)^2 \right]\,dx
\end{eqnarray*}
by Kato's form representation theorem \cite{Kato}. We
find from \eqref{eq:defofB} that
\begin{eqnarray} \label{eq:Green}
\sum_{k}B(u_k, \partial_1 u_0)  \langle u_k,\partial_1 u_0\rangle & = &
[\langle \partial_1 u_0, \Delta \partial_1 u_0\rangle +\Vert \nabla \partial_1
u_0 \Vert^2] = \int_D
\nabla \cdot \left(\partial_1 u_0 \nabla \partial_1 u_0\right) dx, \nonumber
\end{eqnarray}
where Green's identity was applied. This equals
\begin{equation} \label{expression}
\sum_j \left[ \int_{S_j}  (\partial_1 u_0) \partial_j
\partial_1 u_0 dS - \int_{T_j}  (\partial_1 u_0)
\partial_j \partial_1 u_0 dS \right] \ ,
\end{equation}
where $S_j$ and $T_j$ are the faces of $D$ perpendicular to the $j$
direction. Since $\partial_j u_0 \equiv 0$ on $S_j$ and $T_j$ we
find that the sum can be restricted to the indices $j \not= 1$.
Since $\partial_j \partial_1 u_0 = \partial_1 \partial_j u_0$ and $
\partial_j u_0 \equiv 0$ on $S_j$ and $T_j$, the expression
(\ref{expression}) vanishes. After substituting this and the first
order perturbation formula (\ref{eq:fopt}) for $j=1$
into \eqref{eq:E0eq} we arrive at \eqref{eq:bigone}.

\end{proof}

\begin{lemma} \label{zeroderiv}
Assume that the right hand side of (\ref{eq:bigone}) vanishes
for some $a_{1,0} \in \overline{I}$.
Then $E_0(a_{1,0}) = E_0'(a_{1,0})=0$ and $\partial_j \partial_1 u_0(\cdot, a_{1,0}) =0$ on $S_j$ and $T_j$ for
all $j=1,\ldots,d$.

Moreover, $u_0(\cdot,a_{1,0})$ is constant on the right and left faces
$S_1$ and $T_1$ of $D$.
\end{lemma}

It will follow in Theorem~\ref{dichotomy} below that $u_0(\cdot,a_{1,0})$ takes the same value on $S_1$ and $T_1$.

\begin{proof}
The vanishing of the right side of (\ref{eq:bigone}) means that, for $k\not=0$,
\begin{equation} \label{eq:vanish}
\langle \Delta u_k, \partial_1 u_0\rangle  - \langle u_k, \Delta \partial_1 u_0\rangle = 0 \ .
\end{equation}
Using the eigenvalue equation $-\Delta u_k +q_{a_{1,0}} u_k = E_k u_k$ we
can transform this into
\begin{equation}
(E_k-E_0) \langle u_k, \partial_1 u_0\rangle = -\langle u_k, (\partial_1 q_{a_{1,0}}) u_0\rangle \ .
\end{equation}
Now pick any function $f \in C^\infty(D)$ and use the previous
identity to write
\begin{equation} \label{eq:fexp}
\sum_k (E_k-E_0)\langle u_k,   \partial_1 u_0\rangle  \langle f,
u_k\rangle = -\sum_k \langle u_k, ( \partial_1 q_{a_{1,0}})
u_0\rangle \langle f, u_k\rangle + \langle u_0, \partial_1
q_{a_{1,0}} u_0\rangle (f, u_0) \ .
\end{equation}
Since the functions $u_k$,
$k=0,1,2, \ldots$ form an orthonormal basis for $L^2(D)$ and $f, \partial u_0$ are in the form domain we use once more Kato's representation theorem  to obtain
\[
\sum_k (E_k-E_0) \langle u_k, \partial_1 u_0\rangle \langle f,u_k\rangle  =
 \langle \nabla f, \nabla (\partial_1 u_0)\rangle + \langle f,q_{a_{1,0}} \partial_1 u_0\rangle -E_0\langle f,\partial_1 u_0\rangle \ .
\]
Plugging the latter two identities back into (\ref{eq:fexp}) we arrive at
\begin{equation}
\langle \nabla f, \nabla(\partial_1 u_0)\rangle + \langle f,\partial_1 (q_{a_{1,0}} u_0)\rangle - E_0\langle f,\partial_1 u_0\rangle = \langle u_0, (\partial_1 q_{a_{1,0}}) u_0\rangle \langle f,u_0\rangle \ .
\end{equation}
Now pick $f \in C^\infty_c(D)$ so that we can integrate by parts
without boundary terms and get
\begin{equation}
\langle f, \partial_1 (-\Delta u_0 + q_{a_{1,0}} u_0 -E_0u_0)\rangle =
\langle u_0, (\partial_1 q_{a_{1,0}}) u_0\rangle \langle f,u_0\rangle \ .
\end{equation}
Since the left side vanishes for all $f \in C^\infty_c(D)$ and the latter are dense in $L^2(D)$, we
must necessarily have that
$$
 \langle u_0, (\partial_1 q_{a_{1,0}}) u_0\rangle = 0 \ ,
$$
i.e.,  $E_0'(a_{1,0}) = 0$ by (\ref{eq:fopt}).
Thus, we have for all $f \in C^\infty(D)$,
\begin{equation}
\langle \nabla f, \nabla(\partial_1 u_0)\rangle + \langle f,\partial_1 (q_{a_{1,0}} u_0)\rangle - E_0 \langle f,\partial_1 u_0\rangle =0 \ .
\end{equation}
Doing the same integration by parts as above, this time with boundary terms, yields
\begin{equation}
\sum_j \left[\int_{S_j} ( \partial_j \partial_1 u_0) f dS - \int_{T_j}
(\partial_j \partial_1 u_0) f dS \right] = 0
\end{equation}
for all $f \in C^\infty(D)$.  This means that
\begin{equation}
 \partial_j  \partial_1 u_0 = 0
 \end{equation}
 pointwise on  $S_j$ and $T_j$ for all $j =1, \ldots, d$.
 In particular, $\partial_1^2 u_0 (x,a_{0,1}) \equiv 0$ on $S_1$ and $T_1$.
 Since the potential vanishes on the faces $S_1$ and $T_1$ we find that $u_0$ satisfies
the equation
\begin{equation}
-\Delta' u_0 = -\sum_{j=2}^d \frac{\partial^2 u_0}{\partial x_j^2}
= E_0(a_0) u_0
\end{equation}
on $S_1$ and $T_1$.  The function $u_0$ is smooth up to
and including the boundary of $D$, in particular it is a smooth
function on the faces of $D$. Consider the $d-2$ dimensional `edge'
where $S_1$ and $S_2$, say, meet. The gradient of $u_0$ at this
intersection must be of the form
\begin{equation}
\nabla u_0 = \left( 0 , 0 , \partial_3 u_0 , \cdots \partial_d u_0
\right) \ .
\end{equation}
Hence, $u_0$ is an eigenfunction of $-\Delta'$ on $S_1$ and $T_1$
with a Neumann condition on the boundary. Since $u_0$ has a fixed
sign,  $E_0(a_{1,0})$ must be the lowest eigenvalue of the Neumann Laplacian
on $S_1$ and $T_1$ and hence
\begin{equation}
E_0(a_{1,0}) = 0 \ ,
\end{equation}
and $u_0$ is constant on $S_1$ as well as on $T_1$.
 \end{proof}

\begin{theorem}  \label{dichotomy}
Assume that the right side of (\ref{eq:bigone}) vanishes for some
$a^{(0)} \in G$. Then $E_0(a) = 0$ identically in $G$ and for
every $a \in G$ the ground state $u_0(x,a)$ of $H_a$ is constant near the boundary of $D$.
\end{theorem}

\begin{proof}
Write $a^{(0)} = (a_{1,0}, a_{2,0}, \ldots, a_{d,0})$ and consider $q_{a_1}$, $E_0(a_1)$ and $u_0(x,a_1)$ as functions of the first coordinate of $a=(a_1, a_2, \ldots, a_d)$ only, setting $(a_2,\ldots,a_d) = (a_{2,0},\ldots, a_{d,0})$ in Lemmas~\ref{perturbation} and \ref{zeroderiv} above.

We have $a_{1,0}\in (-\delta_1, \delta_2)$ and thus the faces $S_1$ and $T_1$ of $D$ are in the complement of the support of $q_{a_{1,0}}$.

Pick a point  on $S_1$. We may call this point the origin. By the reflection argument already used in the proof of Lemma~\ref{lem:efsmooth} the function $u_0(x,a_{1,0})$ is harmonic in a neighborhood $N$ of $0$. Thus, it has a convergent
power series expansion
\begin{equation}
\sum_\alpha c_\alpha x^\alpha
\end{equation}
which can also be written as
\begin{equation}
\sum_{k=0}^\infty x_1^k \sum_{\alpha'} c_{(k, \alpha')} x^{\alpha'}
=  \sum_{k=0}^\infty x_1^k P_k(x') \ ,
\end{equation}
where $x' = (x_2, \ldots, x_d)$. Since $u_0$ is harmonic we find the
recursion
\begin{equation} \label{harmonic}
\Delta' P_k(x') +(k+1)(k+2)P_{k+2}(x') = 0 \ , k=0,1,2, \ldots \ .
\end{equation}
By Lemma \ref{zeroderiv} we know that $u_0(0,x',a_{1,0})$ is
constant  and we find that $P_0(x')$ is constant, too, and thus
$P_k(x') = 0$ for $k \ge 2$, even. Since $\partial_1 \partial_j
u_0(0,x') = 0$ for $j=2, \ldots, d$ we find that  $\partial_j
P_1(x') = 0$ for $j=2, \ldots, d$ and hence $P_1(x')$ is constant
and thus on account of (\ref{harmonic}), $P_k(x') = 0$ for all $k
\ge 2$. Thus, $u_0(x,a_{0,1}) = c + d x_1$ near $0$,
where $c, d$ are constants.  However, as $u_0(x,a_{0,1})$ satisfies Neumann boundary conditions, we must have $d=0$. This implies that $u_0(x,a_{0,1}) = c$ on $(\alpha_1, \alpha_1+\delta_1+a_{1,0}) \times D'$, where $q_{a_{1,0}}=0$.

In the same way, now starting with a point on the face $T_1$ it is shown that $u_0(x,a_{1,0})$ is constant on $(\beta_1-\delta_2+a_{1,0}, \beta_1) \times D'$.

 If we translate the
potential $q_{a_{1,0}}(x) \to q_{a_{1,0}}(x-(a_1-a_{1,0})) = q(x-a_1)$, the translated
function $v(x) = u_0(x-(a_1-a_{1,0}),a_{1,0})$ is an eigenfunction of the
Schr{\"o}dinger equation $H_{a_1}v = E_0(a_{1,0})v$ and since $v$ has fixed
positive sign, it is the ground state. Thus, the eigenvalue
$E_0(a_1)=0$ identically in $a_1\in I$.

Note that, by assumption in Theorem~\ref{dichotomy}, the above argument can be applied to $E_0(a_1,\ldots,a_d)$ as a function of each one of its variables with the other coordinates fixed. Thus, it vanishes identically in each variable and therefore identically on $G$. Also, it is found that $u_0(x,a^{(0)})$ is constant in a neighborhood of each face $S_j$, $T_j$, $j=1,\ldots,d$. As the union of these neighborhoods is connected, it follows by unique continuation that $u_0(x,a^{(0)})$ is constant near the boundary of $G$. This is then also true for the ground state $u_0(x,a) = u_0(x-(a-a_0), a_0)$ with arbitrary $a\in G$.
\end{proof}

For what follows, it will be convenient to assume that
$(\alpha_1,\beta_1)=(-s,s)$ for some $s>0$, supp$\,q \subset [-r,r]
\times D'$ for some $0<r<s$, and that $q$ is reflection-symmetric
with respect to the first variable, i.e.\ $q(x_1,x')=q(-x_1,x')$ for
all $x_1\in (-s,s)$ and $x'\in D'$.  As a consequence,
the function $a_1 \to E_0(a_1)$ is symmetric about $a_1=0$, in particular
$E'_0(0) = 0$.

\begin{corollary} \label{cor:nonzero}
Let $q$ be smooth and be reflecton symmetric with respect to the first variable. Assume that $E_0(a)$ does not vanish
identically on $G$. Then, for each choice of $(a_2,\ldots, a_d) \in G'$, the function $E_0(a_1) = E_0(a_1,a_2,\ldots, a_d)$ satisfies $E_0'(a_1)<0$ for all $a_1\in (0,s-r]$ and, by
symmetry, $E_0'(0)=0$ and $E_0'(a_1)>0$ for $a_1\in [-(s-r),0)$.
\end{corollary}

\begin{proof}
For fixed $(a_2,\ldots,a_d) \in G'$, denote the right hand side of (\ref{eq:bigone}) by $C(a_1)$. By assumption and Theorem~\ref{dichotomy} we know that $C(a_1)$ is strictly negative for $a_1 \in (-(s-r), s-r)$.
Using the integrating factor $e^{F(a_1)}$, where $F(a_1)$ is an antiderivative of $-4\langle u_0, \partial_1 u_0\rangle(a_1)$,
equation (\ref{eq:bigone}) can be written as
\[
\frac{d}{da_1} \left(e^{F(a_1)}E'_0(a_1)\right) = C(a_1) e^{F(a_1)} \ .
\]
Integrating this equation using the fact that $E_0'(0) = 0$ yields the result. Note that this holds up to the boundary of $[-(s-r),s-r]$ by smoothness of $E_0(a)$ up to the boundary (Lemma~\ref{lem:diff}).
\end{proof}

We are ready to state the strengthened version of Theorem~1.3 of \cite{BLS1}, which will be central to all our later considerations in the localization proof.

The domain $D$ in $\R^d$ may be any
open rectangular parallelepiped and $q$ a $C^{\infty}$-smooth function with
closed support contained in $D$ and having the same symmetry hyper-planes
as $D$. We define $G$ and $E_0(a)$ for $a\in \overline G$ as in
Section~\ref{sec:basicsmooth}. Then $G$ is an open rectangular
parallelepiped and we may assume that its center is in the origin
and $G= (-M_1,M_1) \times \ldots \times (-M_d,M_d)$. Due to the symmetry assumption,
Corollary~\ref{cor:nonzero} applies to each coordinate of $E_0(a)$ and yields
\begin{corollary} \label{cor:offcenterdecay}
Assume that $E_0(a)$ is not identically zero on $G$. For all $a =
(a_1,\ldots,a_d) \in \overline G$ and all $i=1,\ldots,d$ we have
\[ \partial_i E_0(a) \left\{ \begin{array}{ll} < 0, & \mbox{if $a_i >0$}, \\ =0, & \mbox{if $a_i=0$}, \\ >0, & \mbox{if $a_i<0$}. \end{array} \right. \]
\end{corollary}

\noindent {\bf Open Problem:} It is an open problem to show that the phenomenon ``bubbles tend to the corners'', i.e.\ results such as Theorem~1.3 in \cite{BLS1} or Corollary~\ref{cor:offcenterdecay} above, also appears if the rectangular domain $D$ is replaced by more general polyhedra, for example regular $n$-gons in $\R^2$. Again this should need a suitable symmetry assumption on the potential $q$, e.g.\ spherical symmetry. Particularly interesting would be the case of an equilateral triangle, as all other results in this paper would be applicable to this case as well. It was shown in \cite{BLS1} for general smooth convex domains that the minimizing position lies at the boundary.

\section{Lifshitz tails} \label{sec:lifshitz}

In proving smallness of the IDS of $H_{\omega}$ near $E_0$, the almost sure spectral minimum of $H_{\omega}$, we will rely on a prior result of this kind for the RDM (\ref{eq:hamiltonian}), (\ref{eq:potential}) established in \cite{KN2}. Their result requires widely the same assumptions as made above, i.e.\ (\ref{eq:symm}), (\ref{eq:qsupp}), (\ref{eq:rhosupp}) and that $E_0(a)$ does not vanish identically, but needs in addition that the support of the distribution $\mu$ is a {\it finite discrete set}. This violates our assumption (\ref{eq:rhosmooth}) (and every other assumption on $\mu$ which works for our proof of a Wegner estimate in Section~\ref{sec:wegner}).

However, based on the results from Section~\ref{sec:key}, one might expect that configurations in which all bubbles sit in corners, as dealt with in \cite{KN2}, constitute the worst case scenario for smallness of the IDS near $E_0$. We will make this rigorous here by showing that the IDS of $H_{\omega}$ can be bounded from above (up to a multiplicative constant) by the IDS of a modified RDM where all bubbles have been moved to the nearest corner.

More precisely, for an $a\in \overline{G} = [-d_{max}, d_{max}]^d$, let $c(a)$ denote the element of ${\mathcal C}$ closest to $a$. If there is more than one such point, any of them can be chosen; for the sake of definiteness, we may order the points in ${\mathcal C}$ lexicographically and chose the first in the list. For a displacement configuration $\omega = (\omega_i)_{i\in \Z^d} \in \overline{G}^{\Z^d}$ we will write $c(\omega)$ for the {\it closest corner configuration} given by $(c(\omega))_i = c(\omega_i)$, $i\in \Z^d$.

For a non-negative integer $L$, let $\Lambda_{2L+1} = (-L-1/2, L+1/2)^d$ and $H_{\omega,L}$ the restriction of $H_{\omega}$ to $\Lambda_L$ with Neumann boundary conditions. Also, let $\Lambda'_{2L+1} = \Z^d \cap \Lambda_{2L+1}$. We will prove

\begin{proposition} \label{prop:formcomp}
There exists a constant $C\in (0,\infty)$ such that, in the sense of quadratic forms,
\begin{equation} \label{eq:formcomp}
H_{\omega,L}-E_0 \ge \frac{1}{C} (H_{c(\omega),L}-E_0)
\end{equation}
for all $\omega \in \overline{G}^{\Z^d}$ and all $L\ge 0$.
\end{proposition}

We start by showing this for $L=0$, i.e.\ on the level of the single-site operators $H_{\Lambda_1}^N(a)$. This is achieved in the following two lemmas, which separately treat the cases where the bubble is close to a corner or not close.

Recall from Section~\ref{sec:introduction} that $E_0(a)$ denotes the ground state energy of
$H_{\Lambda_1}^N(a)$ for $a\in \overline{G}$ and that the almost sure spectral minimum $E_0$ of $H_{\omega}$ is given by $E_0(a)$ when $a$ is one of the corners ${\mathcal C}$ of $\overline{G}$. As a consequence of Corollary~\ref{cor:offcenterdecay}, $E_0(a)$ grows linearly in the distance of $a$ from the nearest corner, i.e.\ there exists a constant $C\in(0,\infty)$ such that,
\begin{equation}
  \label{eq:1}
  E_0(a)-E_0\geq\frac1C D(a)
\end{equation}
where $D(a) = \min_{c\in {\mathcal C}} |a-c|$.

We will use this to prove
\begin{lemma}
  \label{le:1}
  There exists $C>0$ and $\delta>0$ such that, if $D(a) \le \delta$, then
  \begin{equation}
    \label{eq:3}
    H_{\Lambda_1}^N(a)-E_0\geq \frac1C
    \left(H_{\Lambda_1}^N(c)-E_0+|a-c|)\right).
  \end{equation}
\end{lemma}
\begin{proof}
  Let $a\in \overline{G}$ and pick $c\in {\mathcal C}$ such that $D(a)=|a-c|$. As $q$ is $C^1$, write
  \begin{equation}
    \label{eq:2}
    \begin{split}
      H_{\Lambda_1}^N(a)-E_0&=H_{\Lambda_1}^N(c)-E_0+
      q(\cdot-a)-q(\cdot-c)
      \\&=H_{\Lambda_1}^N(c) - E_0 + (c-a) \cdot \nabla q (\cdot -c) + o(|a-c|).
    \end{split}
  \end{equation}

  Using~\eqref{eq:1}, one obtains that
  \begin{equation*}
    H_{\Lambda_1}^N(c)-E_0+ (c-a) \cdot \nabla q(\cdot -c) \ge \frac{1}{C} |a-c| + o(|a-c|).
  \end{equation*}
  Hence, there exists $\rho\in(0,1)$ small such that, for
  $\sigma\in\mathbb{S}^{d-1}$
  with $a = c+\rho\sigma \in \overline{G}$, one has
  \begin{equation*}
    H_{\Lambda_1}^N(c)-E_0- \rho\,\sigma \cdot \nabla
    q(\cdot-c)\geq\rho/2C.
  \end{equation*}
  We recall Lemma 2.1 of \cite{KN1}: If we suppose $A\geq 0$ and $A+B\geq c_0>0$, then
  we have $A+tB\geq \min(1/2,c_0)\cdot (A+t)$ for $t\in [0,1/2]$, since
  \[
  A+tB=(1-t)A+t(A+B)\geq \tfrac12 A+tc_0\geq \min(\tfrac12, c_0)(A+t).
  \]
Applying this with $A=H_{\Lambda_1}^N(c)-E_0$  and $B=-\rho\sigma\cdot \nabla q(\cdot-c)$,
we learn that there exists $C_\rho = \max(2,2C/\rho) >0$ such that for
  $t\in[0,1/2]$ and for $\sigma \in\mathbb{S}^{d-1}$
  with $c+\rho \sigma \in \overline{G}$, one has
  \begin{equation*}
    H_{\Lambda_1}^N(c)-E_0-t\rho \sigma \cdot\nabla
    q(\cdot-c)\geq\frac1{C_\rho}
    \left(H_{\Lambda_1}^N(c)-E_0+t\right).
  \end{equation*}

  Hence, in view of~\eqref{eq:2}, for $|a-c| \leq \rho/2$ and
  $t=|a-c|/\rho$, one has
  \begin{equation*}
    H_{\Lambda_1}^N(a)-E_0\geq \frac1{C_\rho}
    \left(H_{\Lambda_1}^N(c)-E_0+|a-c|/\rho\right)
    +o(|a-c|).
  \end{equation*}
  Finally, this implies that there exists $\delta>0$ such that for
  $|a-c|\leq\delta$, one has~\eqref{eq:3} and completes the
  proof of Lemma~\ref{le:1}.
\end{proof}

Bubbles which are not close to a corner are easier to handle and considered in
\begin{lemma}
  \label{le:2} Fix $\delta\in(0,1)$. There exists $C=C_\delta\in (0,\infty)$ such
  that, for $D(a)\geq\delta$ and all $c\in {\mathcal C}$, one has
  \begin{equation*}
    H_{\Lambda_1}^N(a)-E_0\geq \frac1C
    \left(H_{\Lambda_1}^N(c)-E_0+|a-c|\right).
  \end{equation*}
\end{lemma}
\begin{proof}
  Fix $\delta\in(0,1)$. By (\ref{eq:1}) there exists $\eta>0$
  such that, for $a$ with $D(a)\geq\delta$,
  $H_{\Lambda_1}^N(a)-E_0\geq\eta$. Hence, as
  $|q(x-a)-q(x-c)| \geq-2\|q\|_\infty$,
  there exists $C>1$ such that, for $a$ satisfying $D(a)\geq\delta$ and all $c\in {\mathcal C}$,
  one has
  \begin{equation*}
    (C+1)(H_{\Lambda_1}^N(a)-E_0)-(H_{\Lambda_1}^N(c)-E_0)
    \geq C\eta-2\|q\|_\infty\geq \eta\geq \frac1C|a-c|.
  \end{equation*}
  Hence,
  \begin{equation*}
    \begin{split}
      H_{\Lambda_1}^N(a)-E_0&\geq\frac1{C+1}
      \left(H_{\Lambda_1}^N(c)-E_0+\frac1C|a-c|\right)\\
      &\geq\frac1{C(C+1)}\left(H_{\Lambda_1}^N(c)-E_0+
        |a-c|\right).
    \end{split}
  \end{equation*}
  This proves Lemma~\ref{le:2}.
\end{proof}

To complete the {\em proof of Proposition~\ref{prop:formcomp}} we
need to extend the result of the previous lemmas to general boxes
$\Lambda_{2L+1}$. This is done by an argument previously used in the
proof of Theorem~2.1 in \cite{KN1} which makes crucial use of
properties of Neumann boundary conditions.

For $\psi\in H^1(\Lambda_{2L+1})$, the form domain of $H_{\omega,L}$, one has that the restriction of $\psi$ to $\Lambda_1(i)$ is in $H^1(\Lambda_1(i))$ for each $i\in \Lambda'_{2L+1}$ and
\begin{equation*}
  \langle(H_{\omega,L}-E_0)\psi,\psi\rangle=
  \sum_{i\in \Lambda'_{2L+1}}\langle(-\Delta-E_0+q(\cdot-i-\omega_i)\psi,
  \psi\rangle_{\Lambda_1(i)}
\end{equation*}
where $\langle\cdot,\cdot\rangle_{A}$ denotes the standard scalar
product in $L^2(A)$. This may also be applied to the modified displacement model $H_{c(\omega)}$,
\begin{equation*}
  \langle(H_{c(\omega),L}-E_0)\psi,\psi\rangle=
  \sum_{i\in \Lambda'_{2L+1}}\langle(-\Delta-E_0+q(\cdot-i-c(\omega_i))\psi,
  \psi\rangle_{\Lambda_1(i)}.
\end{equation*}
Hence, using Lemmas~\ref{le:1} and~\ref{le:2} on each term in the sum, we obtain that
\begin{equation*}
  \begin{split}
  \langle(H_{\omega,L}-E_0)\psi,\psi\rangle&\geq\frac1C
  \sum_{i\in \Lambda'_{2L+1}}\langle(H_{\Lambda_1}^N(c(\omega)_i)-E_0+|\omega_i-c(\omega_i)|) \psi, \psi \rangle_{\Lambda_1(i)}
  \\&\geq \frac{1}{C} \langle(H_{c(\omega),L}-E_0)\psi,\psi\rangle,
  \end{split}
\end{equation*}
where the positive term $\sum_i \langle |\omega_i - c(\omega_i)|\psi, \psi \rangle_{\Lambda_1(i)}$ was omitted. This completes the proof of Proposition~\ref{prop:formcomp}.

The random displacement model $
H_{c(\omega)}=-\Delta+\sum_{i\in\Z^d}q(\cdot-c(\omega_i))$ has i.i.d.\ displacement vectors $(c(\omega))_{i\in\Z^d}$, whose distribution is discrete with support given by the corners ${\mathcal C}$ of $\overline{G}$. Thus, it satisfies the assumptions of Theorem 4.1 in~\cite{KN2}; hence, by the proof of Theorem~1.2 in \cite{KN2}, in particular (3.2) in the same work, there exist $C_1>0$, $C_2<\infty$ and $\mu>1$ such that, for all $L$,
\begin{equation*}
  \PP(H_{c(\omega),L}\text{ has an eigenvalue less than
  }E_0+C_1/L^2)\leq C_2 L^d \mu^{-L}.
\end{equation*}

Note that this requires that $d \ge 2$. The argument leading to Theorem~4.1 of \cite{KN2} uses
crucially the uniqueness (up to translations) of the minimizing configuration of the potentials
proved in \cite{BLS2}, which holds only for $d \ge 2$.  See Section~\ref{sec:discussion} below for a comment on the differences for $d=1$. The assumption $r<1/4$ in (\ref{eq:qsupp}) was used in \cite{BLS2} for a more technical reason (rather than just $r<1/2$)  and thus also enters our argument here.

Using Proposition~\ref{prop:formcomp}, we immediately obtain the following finite volume bound on the probability for finding low lying eigenvalues. It is this result which enters the proof of localization via multiscale analysis.
\begin{corollary}
  \label{le:3}
  There exist $C_1>0$, $C_2<\infty$ and $\mu>1$ such that, for all $L$,
  \begin{equation} \label{eq:Lifshitzbound}
    \PP( H_{\omega,L}\text{ has an eigenvalue less than
    }E_0+C_1/L^2)\leq C_2 L^d \mu^{-L}.
  \end{equation}
\end{corollary}
The results of \cite{KN2} also show that the integrated
density of states of $H_{c(\omega)}$, say $\tilde N$,
satisfies a Lifshitz tails estimate of the form
\begin{equation} \label{eq:tildetail}
  \limsup_{\substack{E\to E_0\\E> E_0}}\frac{\log|\log\tilde
    N(E)|}{\log(E-E_0)}\leq -\frac12.
\end{equation}
By Proposition~\ref{prop:formcomp} we see that $N$, the integrated density of states
of $H_\omega$, satisfies, for $E\geq E_0$,
\begin{equation*}
 N(E)\leq\tilde N(E_0+C(E-E_0)).
\end{equation*}
Hence, we have proven the following result which is not required in the localization proof but stated here for its independent interest.
\begin{theorem}
  \label{thr:1}
 The IDS $N$ of $H_{\omega}$ has a Lifshitz tail of the form
  \begin{equation*}
    \limsup_{\substack{E\to E_0\\E> E_0}}\frac{\log|\log
      N(E)|}{\log(E-E_0)}\leq -\frac12.
  \end{equation*}
\end{theorem}

We expect that the Lifshitz exponent $1/2$ is not optimal and should instead be $d/2$, the standard value known from the Anderson model. We can think of two ways in which one could try to get this improvement, none of which we know how to make rigorous. One approach would be to show that (\ref{eq:tildetail}) holds with exponent $d/2$, which would immediately give the same in Theorem~\ref{thr:1}. That this should hold is discussed in \cite{KN2}, where the exponent $1/2$ is found due to one part of the proof which uses an essentially one-dimensional argument.

Another way to argue would be to make use of the term $\sum_i |\omega_i - c(\omega_i)| \chi_{\Lambda_1(i)}$ which was dropped in the proof of Proposition~\ref{prop:formcomp}. Under our assumptions (\ref{eq:rhosupp}) and (\ref{eq:rhosmooth}) this means that one would have to show standard Lifshitz tails for an Anderson-type model where the unperturbed operator is the random operator $H_{c(\omega)}$. However, the known methods do not work for the irregular background potential appearing here.

\section{A Wegner estimate} \label{sec:wegner}

Throughout this section, we write $H_i(a)=H_{\Lambda_1(i)}^N(a)$ for simplicity.
Our goal in this section if to prove the following Wegner estimate for energies near
$E_0=\inf_{a\in\overline{G}}E_0(a)$, where $E_0(a)=\inf\sigma(H_i(a))$:

\begin{theorem}
  \label{thm:wegner}
  There exists $\delta>0$ such that, for any $\alpha\in(0,1)$, there
  exists $C_\alpha>0$ such that, for every interval $I\subset
  [E_0,E_0+\delta]$ and $L\in \N$,
  \begin{equation} \label{eq:wegnerest}
    \E(\tr\chi_I(H_{\omega,L})) \le C_\alpha |I|^\alpha L^{d}.
  \end{equation}
\end{theorem}
\noindent By Chebychev's inequality this implies that for every interval $[E-\eta,E+\eta] \subset [E_0, E_0+\delta]$,
\begin{equation} \label{eq:Wegnercl}
\PP(\mbox{dist}(\sigma(H_{\omega,L},E)) \le \eta) \le C_{\alpha}' \eta^{\alpha} L^e,
\end{equation}
the more classical form of the Wegner estimate used in applications.
As a consequence of the existence of the integrated density
of states (see e.g.~\cite{Stollmann}) and Theorem~\ref{thm:wegner}, we also
get
\begin{corollary}
  \label{cor:1}
  There exists $\delta>0$ such that, for any $\alpha\in(0,1)$, the
  integrated density of states of $H_\omega$ is $\alpha$-H{\"o}lder
  continuous in $[E_0,E_0+\delta]$.
\end{corollary}

The rest of this section will be devoted to the proof of Theorem~\ref{thm:wegner}.

For a function $f$ on $G$ we set
\[
(\partial_c f)(a) :=  \frac{c(a)-a}{|c(a)-a|} \cdot \nabla f(a),
\]
with $c(a)$ denoting the corner closest to $a$ as in Section~\ref{sec:lifshitz}. Thus, $\partial_c$ denotes the directional derivative in the direction of the closest corner, where points $a$ with multiple closest corners will not play a role in the arguments below (starting from (\ref{eq:cutoff}) below we introduce a cut-off which restricts the values of $a$ relevant for the proof to small neighborhoods of the corners).

By Corollary~\ref{cor:offcenterdecay} there exist $\delta_0>0$ and $r_0>0$ such that
\begin{equation} \label{eq:boundonderiv}
\partial_{c}E_0(a) \leq -\delta_0 \quad \mbox{for all $a\in A_{r_0} := \{a \in G: |c(a)-a| \leq r_0\}$},
\end{equation}
a neighborhood of $\mathcal{C}$.

Let $\eta\in C^\infty(\mathbb{R})$ such that $0\le \eta \le 1$, $\eta(r)=1$ for $r\leq r_0$ and
$\eta(r)=0$ for $r \geq 2r_0$. Using this function as a cut-off, we localize the vector fields associated with $\partial_c$ onto a neighborhood of the corners, defining
\begin{equation} \label{eq:cutoff}
 (\partial'_c f)(a) := \eta(|c(a)-a|) (\partial_c f)(a).
\end{equation}

For each $i\in \mathbb{Z}^d$, we write
\begin{equation} \label{eq:Hamderiv}
\partial'_{c,\omega_i}H_\omega
=\partial'_{c,\omega_i} q(\cdot-i-\omega_i)
=-\eta(|c(\omega_i)-\omega_i|)
\frac{c(\omega_i)-\omega_i}{|c(\omega_i)-\omega_i|}  \cdot(\nabla q)(\cdot -i-\omega_i).
\end{equation}

If $\psi \in H^1(\Lambda_{2L+1})$, the form domain of $H_{\omega,L}$, then $\psi_i := \psi|_{\Lambda_1(i)} \in H^1(\Lambda_1(i))$, the form domain of $H_i(\omega_i)$, and, with the usual abuse of notation for the quadratic form,
\begin{equation} \label{eq:formdecomp}
\langle \psi, H_{\omega,L} \psi \rangle = \sum_{i\in \Lambda_{2L+1}'} \langle \psi_i, H_i(\omega_i) \psi_i \rangle,
\end{equation}
as well as
\begin{equation} \label{eq:derlocal}
\sum_{i\in \Lambda_{2L+1}'} \langle \psi, \partial'_{c,\omega_i} H_{\omega,L} \psi \rangle = \sum_{i\in \Lambda_{2L+1}'} \langle \psi_i, \partial'_{c,\omega_i} H_i(\omega_i) \psi_i \rangle.
\end{equation}

\begin{proposition} \label{prop:key}
There exist $\delta_1>0$ and $\delta_2>0$ such that
\begin{equation} \label{eq:key}
-\sum_{i\in \Lambda_{2L+1}'} \langle \psi, (\partial'_{c,\omega_i}
H_{\omega,L}) \psi \rangle \ge \delta_1 \|\psi\|^2
\end{equation}
for all $L\in \N$, and $\psi \in H^1(\Lambda_{2L+1})$ with $\langle \psi, (H_{\omega,L}-E_0)\psi \rangle \le \delta_2 \|\psi\|^2$.
\end{proposition}

In the proof of Theorem~\ref{thm:wegner} at the end of this section, Proposition~\ref{prop:key} provides the crucial technical result which will allow us to link our argument with previously known strategies for proving Wegner estimates. Thus, we pause here to motivate the origin of Proposition~\ref{prop:key}, comparing with the situation of an Anderson-type model
\[ H_{\omega}^A = -\Delta + \sum_{i \in \Z^d} \lambda_i q(x-i) \]
with random coupling constants $\lambda_i$ and suitable non-negative single-site potential $q$. If $H_{\omega,L}^A$ is the restriction of $H_{\omega}^A$ to $\Lambda_{2L+1}$ with appropriate boundary condition, $I\subset \R$ a compact interval, and $\{\phi_j\}$ all normalized eigenfunctions to eigenvalues of $H_{\omega,L}^A$ in $I$, then Proposition~\ref{prop:key} should be considered an analogue to the fact that
\[ \sum_{i \in \Lambda_{2L+1}'} \langle \phi_j, q(\cdot -i) \phi_j \rangle \ge C_0 >0 \]
uniformly in $j$ and $L$, which is a key step in the proof of a Wegner estimate for $H_{\omega}^A$ in $I$, e.g.\ \cite{CHN}.

For this, first note that, in a sense as in (\ref{eq:Hamderiv}), $\partial H_{\omega,L}^A / \partial \lambda_i = q(\cdot-i)$. In the RDM the random parameters $\omega_i$ are vector-valued, which gives us the choice to differentiate in (\ref{eq:key}) with respect to a suitably chosen vector field, for which we take for every $\omega_i$ the partial derivative in the direction of the closest corner. Finally, as indicated by the prime in (\ref{eq:key}), we need to sum only over those $i$ where $\omega_i$ is close to a corner. This is plausible by the fact that we establish (\ref{eq:key}) only for $\psi$ with energy close to $E_0$, which forces them to have most of their mass in cubes where $\omega_i$ is close to a corner, see Lemma~\ref{lem2} below.

The proof of Proposition~\ref{prop:key} will be prepared by a series of lemmas.

Let $P_i(a)$ be the eigenprojection onto the groundstate of $H_i(a)$, $\overline{P}_i(a) := I-P_i(a)$, and
\[ E_1 := \inf_{a\in \overline{G}}  (\sigma(H_i(a)) \setminus \{E_0(a)\} ) > E_0.\]

\begin{lemma} \label{lem1}
If $\psi \in H^1(\Lambda_{2L+1})$ with $\langle \psi,(H_{\omega,L}-E_0) \psi \rangle \le \delta_2 \|\psi\|^2$, then
\begin{equation} \label{eq:lem1}
\sum_{i\in \Lambda_{2L+1}'} \|\overline{P}_i (\omega_i)\psi_i\|^2
\le \frac{\delta_2}{E_1-E_0} \|\psi\|^2.
\end{equation}
\end{lemma}

\begin{proof}
We have
\begin{eqnarray*}
\langle \psi, (H_{\omega,L}-E_0)\psi \rangle & = & \sum_{i\in \Lambda_{2L+1}'} \left\{ (E_0(\omega_i)-E_0) \|P_i(\omega_i) \psi_i\|^2 \right. \\
& & \mbox{} \left. + \langle \overline{P}_i(\omega_i) \psi_i, (H_i(\omega_i)-E_0) \overline{P}_i(\omega_i) \psi_i \rangle \right\} \\
& \ge & (E_1-E_0) \sum_{i\in \Lambda_{2L+1}'} \|\overline{P}_i(\omega_i) \psi_i\|^2,
\end{eqnarray*}
which yields (\ref{eq:lem1}) by the assumption.
\end{proof}

By the results of \cite{BLS1} we have
\[ E_{0,r_0} := \inf_{a\in \overline{G}\setminus A_{r_0}} E_0(a) > E_0.\]

Define $\Lambda_{2L+1}'' := \{ i\in \Lambda_{2L+1}' : \;\omega_i \not\in A_{r_0} \}$.

\begin{lemma} \label{lem2}
If $\psi \in H^1(\Lambda_{2L+1})$ with $\langle \psi, (H_{\omega,L}-E_0) \psi \rangle \le \delta_2 \|\psi\|^2$, then
\begin{equation} \label{eq:lem2}
\sum_{i\in \Lambda_{2L+1}''} \|\psi_i\|^2 \le \frac{\delta_2}{E_{0,r_0}-E_0} \|\psi\|^2.
\end{equation}
\end{lemma}

\begin{proof}
This follows from the assumption and
\[ \langle \psi, (H_{\omega,L}-E_0) \psi \rangle \ge \sum_{i\in \Lambda_{2L+1}''} (E_{0,r_0}-E_0) \|\psi_i\|^2.\]
\end{proof}

\begin{lemma} \label{lem3}
There exist $C_1<\infty$ and $C_2 <\infty$ such that for $\varphi
\in L^2(\Lambda_1(i))$,
\begin{equation} \label{eq:lem3}
-\langle \varphi, (\partial'_{c,a} H_i(a)) \varphi \rangle \ge
-(\partial'_{c,a} E_0(a)) \|P_i(a)\varphi\|^2 - C_1 \|\varphi\|
\|\overline{P}_i \varphi\| - C_2 \|\overline{P}_i \varphi\|^2.
\end{equation}
\end{lemma}

\begin{proof}
Omitting the variable $a$ we have,
\[ \partial'_c P_i = -\partial'_c \overline{P}_i = -\partial'_c
(\overline{P}_i^2) = \overline{P}_i (\partial'_c P_i) + (\partial'_c
P_i) \overline{P}_i \] and
\begin{eqnarray*}
\partial'_c H_i & = & \partial'_c (E_0 P_i + \overline{P}_i H_i
\overline{P}_i) \\
& = & (\partial'_c E_0) P_i + E_0 (\partial'_c P_i) + (\partial'_c
\overline{P}_i) H_i \overline{P}_i + \overline{P}_i (\partial'_c H_i)
\overline{P}_i + \overline{P}_i H_i (\partial'_c \overline{P}_i) \\
& = & (\partial'_c E_0) P_i + E_0 (\partial'_c P_i) \overline{P}_i +
E_0 \overline{P}_i (\partial'_c P_i) \\
& & \mbox{} + \overline{P}_i (\partial'_c H_i) \overline{P}_i -
(\partial'_c P_i) H_i \overline{P}_i - \overline{P}_i H_i
(\partial'_c P_i).
\end{eqnarray*}
Thus, for $\varphi \in L^2(\Lambda_1(i))$,
\begin{eqnarray*}
-\langle \varphi, (\partial'_c H_i) \varphi \rangle & = & -(\partial'_c
E_0) \langle \varphi, P_i \varphi \rangle - 2 E_0 \re \langle
\overline{P}_i \varphi, (\partial'_c P_i) \varphi \rangle \\
& & \mbox{} - \langle \overline{P}_i \varphi, (\partial'_c H_i)
\overline{P}_i \varphi \rangle +2 \re \langle H_i (\partial'_c P_i)
\varphi, \overline{P}_i \varphi \rangle.
\end{eqnarray*}
This implies (\ref{eq:lem3}) if we can show that the operators
$\partial'_{c,a} H_i(a)$, $(\partial'_{c,a} P_i)(a)$ and $H_i(a)
(\partial'_{c,a} P_i)(a)$ are bounded in $L^2(\Lambda_1(i))$,
uniformly in $a\in \overline{G}$. This is clear for
\[
\partial'_{c,a} H_i(a) = -\eta(|c(a)-a|) \frac{c(a)-a}{|c(a)-a|} \cdot (\nabla q)(\cdot-i-a).
\]
For $P_i(a)$
we write, as in Section~\ref{sec:key},
\[ P_i (a) = \frac{1}{2\pi i} \oint_C (H_i(a)-z)^{-1}\, dz,\]
where $C$ circles around $E_0(a)$, and can be chosen locally
independent of $a$ and with distance to $E_0(a)$ which is bounded
below uniformly in $a\in \overline{G}$ (as the distance of $E_0(a)$
and $E_1(a)$ is uniformly bounded below). From this we can conclude
that
\[ (\partial'_{c,a} P_i)(a) = \frac{1}{2\pi i} \oint_C
(H_i(a)-z)^{-1} (\partial'_{c,a} H_i(a)) (H_i(a)-z)^{-1}\,dz \] is
uniformly bounded in $a$. In this expression one can absorb an
additional factor $H_i(a)$ to show uniform boundedness of $H_i(a)
(\partial'_{c,a} P_i)(a)$.

\end{proof}

We are now ready to complete the proof of
Proposition~\ref{prop:key}:

\begin{proof}
By Lemmas~\ref{lem1} and \ref{lem2} we have
\begin{eqnarray*}
\sum_{i\not\in \Lambda_{2L+1}''} \|P_i(\omega_i) \psi_i\|^2 & = &
\sum_{i \not\in \Lambda_{2L+1}''} (\|\psi_i\|^2 -
\|\overline{P}_i(\omega_i) \psi_i\|^2) \\
& \ge & \|\psi\|^2 - \sum_{i\in \Lambda_{2L+1}''} \|\psi_i\|^2 -
\sum_{i\in \Lambda_{2L+1}'} \|\overline{P}_i(\omega_i) \psi_i\|^2 \\
& \ge & \left( 1-\frac{\delta_2}{E_1-E_0}
-\frac{\delta_2}{E_{0,r_0}-E_0} \right) \|\psi\|^2.
\end{eqnarray*}
Let $c_3 := \sup_{a\in G} |\partial'_{c,a} E_0(a)|$. With
Lemma~\ref{lem3} we find
\begin{eqnarray}
\lefteqn{- \sum_{i\in \Lambda_{2L+1}'} \langle \psi, (\partial'_{c,\omega_i}
H_{\omega,L})\psi \rangle \;\; = \;\; - \sum_{i\in \Lambda_{2L+1}'}
\langle \psi_i, (\partial'_{c,\omega_i} H_i(\omega_i)) \psi_i \rangle} \nonumber \\
& \ge & \sum_{i\in \Lambda_{2L+1}'} (\partial'_{c,\omega_i}
E_0(\omega_i)) \|P_i(\omega_i) \psi_i\|^2 \nonumber \\
& & \mbox{} - C_1 \sum_{i \not\in \Lambda_{2L+1}'} \|\psi_i\|
\|\overline{P}_i(\omega_i) \psi_i\| - C_2 \sum_{i\in
\Lambda_{2L+1}'} \|\overline{P}_i(\omega_i) \psi_i\|^2 \nonumber \\
& \ge &  -\sum_{i \not\in \Lambda_{2L+1}''} (\partial'_{c,\omega_i}
E_0(\omega_i)) \|P_i(\omega_i)\psi_i\|^2 - \sum_{i\in
\Lambda_{2L+1}''} (\partial'_{c,\omega_i} E_0(\omega_i))
\|P_i(\omega_i) \psi_i\|^2 \nonumber \\
& & \mbox{} -\frac{C_1}{2} \|\psi\| \left( \sum_{i\in
\Lambda_{2L+1}'} \|\overline{P}_i(\omega_i) \psi_i\|^2 \right)^{1/2}
- C_2 \sum_{i\in \Lambda_{2L+1}'} \|\overline{P}_i(\omega_i)
\psi_i\|^2.
\end{eqnarray}

Now we use that $-\partial'_{c,\omega_i} E_0(\omega_i)$ is bounded from below by $-c_3$ if $i\in
\Lambda_{2L+1}''$ and by $\delta_0$ if $i\not\in \Lambda_{2L+1}''$ (the latter means $\omega_i \in A_0$ and thus $\partial'_{c,\omega_i} E_0(\omega_i) = \partial_{c,\omega_i} E_0(\omega_i)$, so we can use (\ref{eq:boundonderiv})). Also using the bounds from Lemma~\ref{lem1} and
\ref{lem2} again, we arrive at
\begin{eqnarray*}
-\sum_{i\in \Lambda_{2L+1}'} \langle \psi, (\partial'_{c,\omega_i}
H_{\omega,L})\psi \rangle & \ge & \delta_0 \left(
1-\frac{\delta_2}{E_1-E_0} -\frac{\delta_2}{E_{0,r_0}-E_0} \right)
\|\psi\|^2 \\
& & \mbox{} - c_3 \frac{\delta_2}{E_{0,r_0}-E_0} \|\psi\|^2 -
\frac{C_1}{2} \sqrt{\frac{\delta_2}{E_1-E_0}} \|\psi\|^2 \\
& & \mbox{} - C_2 \frac{\delta_2}{E_1-E_0} \|\psi\|^2.
\end{eqnarray*}
Choose $\delta_2 > 0$ such that
\[ \frac{c_3 \delta_2}{E_{0,r_0}-E_0} + \frac{C_1}{2}
\sqrt{\frac{\delta_2}{E_1-E_0}} + \frac{C_2 \delta_2}{E_1-E_0} <
\frac{\delta_0}{4},\]
and
\[ \frac{\delta_2}{E_1-E_0} + \frac{\delta_2}{E_{0,r_0}-E_0} <
\frac{\delta_0}{2}. \]
Then
\[  -\sum_{i\in \Lambda_{2L+1}'} \langle \psi,
(\partial'_{c,\omega_i} H_{\omega,L}) \psi \rangle \ge
\frac{\delta_0}{4} \|\psi\|^2.\]

This proves (\ref{eq:key}) with $\delta_1 = \delta_0/4$.

\end{proof}

We now prove Theorem~\ref{thm:wegner}. We follow the approach
developed in \cite{HK} based on $L^p$ estimates of the spectral shift
function (see also \cite{CHN}). The method can be adapted to our model thanks
to Proposition~\ref{prop:key}.

For $\delta_2$ from Proposition~\ref{prop:key} choose $\delta =
\delta_2/2$ and let $I\subset [E_0,E_0+\delta]$ be an interval
of the form $I=[E-\varepsilon,E+\varepsilon]$.

Let $\chi \in C^{\infty}(\R)$ be a real-valued function such that $\chi(x)=-1$
for $x\le -\varepsilon$; $\chi(x)=0$ for $x\ge \varepsilon$;
$\chi'\geq0$; and $\|\chi'\|_{\infty} \le 1/\varepsilon$.

By our assumption on $I$ and the Gohberg-Krein formula, see e.g.\
Proposition~2 in \cite{Simon98}, Proposition~\ref{prop:key} implies
that
\begin{equation*}
  \begin{split}
    \tr \left(-\sum_{i\in \Lambda_{2L+1}'} \partial'_{c,\omega_i}[
      \chi(H_{\omega,L}-E+t)]\right)&=\tr\left(\chi'(H_{\omega,L}-E+t)\left(
      -\sum_{i\in \Lambda_{2L+1}'} (\partial'_{c,\omega_i} H_{\omega})
    \right)\right)\\&\geq \delta_1\tr\left(\chi'(H_{\omega,L}-E+t)\right).
  \end{split}
\end{equation*}
Then, as $\supp \chi' \subset [-\varepsilon,\varepsilon]$ and as
$\chi'\geq0$, one has
\begin{eqnarray}
 \label{eq:wegner1}
 \esp(\tr \chi_I(H_{\omega,L})) & \le & \esp(\tr
 \int_{-2\varepsilon}^{2\varepsilon} \chi'(H_{\omega,L}-E+t)\,dt) \nonumber \\
 & \leq &\frac1{\delta_1} \sum_{i\in \Lambda_{2L+1}'}
 \int_{-2\varepsilon}^{2\varepsilon}\esp\left(\tr\left(-\partial'_{c,\omega_i}[
     \chi(H_{\omega,L}-E+t)]\right)\right)\,dt.
\end{eqnarray}

In the above expectation we want to write the integration with respect to $\omega_i$ over $G$ as a sum of integrals over the intersection of $G$ with each one of the $2^d$ orthants, using polar coordinates with respect to the corners $c\in {\mathcal C}$ in each orthant. For this we represent $a\in G$ by $(r, \theta, c(a)) \in (0,\infty) \times \mathbb{S}^{d-1} \times {\mathcal C}$, where $c(a)$ again denotes the corner closest to $a$ and $(r,\theta)$ polar coordinates of $a-c(a)$. For a function $f$ supported near the corners ${\mathcal C}$ this means that
\[ \int_G f(a) \rho(a)\,da = \sum_{c\in {\mathcal C}} \int_{\mathbb{S}_c^{d-1}} \int_0^{\infty} f(r\theta -c) \rho(r\theta -c) r^{d-1}\,dr\,d\theta, \]
with $\mathbb{S}_c^{d-1}$ denoting the intersection of $\mathbb{S}^{d-1}$ with the orthant containing $-c$.

With $a=\omega_i$ this leads to
\begin{multline}\label{eq:4}
\mathbb{E}\bigl(\tr\bigl(-\partial'_{c,\omega_i}\bigl[
\chi(H_{\omega_,L}-E+t)\bigr]\bigr)\bigr) \\
= \hat{\E}_i \biggl(\tr \sum_{c\in {\mathcal C}} \int_{\mathbb{S}_c^{d-1}} \int_0^{2r_0} \partial_r [ \chi(H_{\omega,L}-E+t)] \eta(r) \rho(r\theta -c) r^{d-1} \,dr\,d\theta \biggr),
\end{multline}
where  $\hat{\E}_i$
denotes the expectation with respect to the random variables $(\omega_j)_{j\not= i}$. Here a sign-change is due to the fact that $\partial_r$ acts in the direction opposite to $\partial_{c,\omega_i}$. By integration by parts, we have
\begin{multline}\label{eq:6}
\int_0^{2r_0} \partial_r \bigl[ \chi(H_{\omega,L}-E+t)\bigr] \eta(r)\rho(r\theta-c) r^{d-1} dr \\
= - \int_0^{2r_0} \bigl[\chi(H_{\omega,L}-E+t)-\chi(H_{\omega^{(i)},L}-E+t)
\bigr] \partial_r (\eta(r)\rho(r\theta-c)r^{d-1})dr,
\end{multline}
where $\omega^{(i)}$ is the random variable such that
$\omega^{(i)}_j=\omega_j$ for $j\neq i$, and
$\omega^{(i)}_i=c(\omega_i)$. Note that the second term in (\ref{eq:6}) actually integrates out to zero, as $\eta(r) r^{d-1}$ vanishes at both endpoints. But we include this term in the integral to be able to make use of bounds on $\chi(H_{\omega,L}-E+t)-\chi(H_{\omega^{(i)},L}-E+t)$.

Now~\eqref{eq:6} implies that
\begin{multline} \label{eq:4b}
  \tr\left(\int_0^{2r_0}
    \partial_{r}[
    \chi(H_{\omega,L}-E+t)]\eta(r)\rho(r\theta-c)r^{d-1}dr\right)
  \\=-\int_0^{2r_0}\left(\int_\R\xi ( \cdot \; ;
  H_{\omega,L} , H_{\omega^{(i)},L} )\chi'(\lambda)d\lambda\right)
  \partial_r[\eta(r)\rho(r\theta-c)r^{d-1}]dr
\end{multline}
This uses the spectral shift function $\xi ( \lambda \; ; H_{\omega,L} ,
H_{\omega^{(i)},L} )$ for the pair $( H_{\omega,L} , H_{\omega^{(i)},L} )$ which is
defined so that
\begin{equation*}
  \int_\R\varphi'(\lambda)\xi  ( \lambda \; ; H_{\omega,L} ,
  H_{\omega^{(i)},L} )d\lambda=\tr(\varphi(H_{\omega,L})-
  \varphi(H_{\omega^{(i)},L})).
\end{equation*}
for all $\varphi \in \mathcal{C}_0^{\infty}(\R)$. The invariance principle for the spectral shift function (see
e.g.~\cite{Yafaev} or \cite{BY}) states that the spectral shift function $\xi (
\lambda \; ; H_{\omega,L} , H_{\omega^{(i)},L} )$ can be written as
\begin{equation*}
  \xi ( \lambda \; ;  H_{\omega,L} , H_{\omega^{(i)},L} ) = -
  \xi ( g_k ( \lambda ) \; ; g_k (  H_{\omega,L} ), g_k(H_{\omega^{(i)},L} ) ) .
\end{equation*}
Here we define $g_k(\lambda)=(\lambda+M)^{-k}$ and  $M$ is picked such that
$\omega$-a.s., $\inf_xV_\omega(x)\geq -M+1$.

By definition
$H_{\omega,L}-H_{\omega^{(i)},L}=q(\cdot-i-\omega_i)-q(\cdot-i-c(\omega_i))$ which is
a bounded, compactly supported potential.   In Section 5 of~\cite{HK}, it is
proved that (actually for more general operators), if $k>pd/2+1$ and
$p>1$, the operator $g_k(H_{\omega,L})-g_k(H_{\omega^{(i)},L})$ is super
trace class of order $1/p$, i.e. its singular values to the power $1/p$
are summable; the $p$-th power of this sum is denoted by
$\|\cdot\|_{1/p}$.  Moreover,
$\|g_k(H_{\omega,L})-g_k(H_{\omega^{(i)},L})\|_{1/p}\leq C_0$ for $C_0>0$
independent of $\omega$ and $L$.

Using a simple change of variables and the bound $\|\xi(\cdot; A, B)\|_{L^p} \le \|A-B\|_{1/p}^{1/p}$ proven in \cite{CHN} we find
\begin{eqnarray} \label{eq:7}
\int_{-\varepsilon}^{\varepsilon} |\xi(\lambda;H_{\omega,L}, H_{\omega^{(i)},L})|^p \,d\lambda
& = & \int_{-\varepsilon}^{\varepsilon} |\xi(g_k(\lambda); g_k(H_{\omega,L}), g_k(H_{\omega^{(i)},L}))|^p\, d\lambda \nonumber \\
& \le & C \int_{\R} |\xi(s; g_k(H_{\omega,L}), g_k(H_{\omega^{(i)},L}))|^p\, ds \nonumber \\
& \le & C \|g_k(H_{\omega,L}) - g_k(H_{\omega^{(i)},L})\|_{1/p} \nonumber \\
& \le & CC_0.
\end{eqnarray}
As $\|\chi'\|_{\infty} \le 1/\varepsilon$ and
supp$\chi'\subset[-\varepsilon,\varepsilon]$,~\eqref{eq:7} and the
H{\"o}lder inequality imply that, for any $q\in(1,+\infty)$, there exists
$C_q>0$ such that
\begin{equation*}
  \sup_\omega\left|\int_\R\xi ( \cdot \; ;
    H_{\omega,L} , H_{\omega^{(i)},L} )\chi'(\lambda)d\lambda\right|
  \leq C_q\,\varepsilon^{1/q-1}.
\end{equation*}
As $\partial_r[\eta(r)\rho(r\theta-c)r^{d-1}]$ is a bounded, compactly supported
function uniformly in $\theta$ and $c$, \eqref{eq:4b} then implies that
\begin{equation*}
  \sup_\omega\left|\tr\left(\int_0^{2r_0}\partial_{r}[
      \chi(H_{\omega,L}-E+t)]\eta(r)\rho(r\theta-c)r^{d-1} dr\right) \right|
  \leq C_q\,\varepsilon^{1/q-1}.
\end{equation*}
Plugging this into~\eqref{eq:4} and then into~\eqref{eq:wegner1}, we
get that, for any $q\in(1,+\infty)$, there exists $C_q>0$ such that
\begin{equation*}
  \esp(\tr \chi_I(H_{\omega,L})) \leq
  C_q\,|\Lambda_{2L+1}'|\,\varepsilon^{1/q-1}
  \int_{-2\varepsilon}^{2\varepsilon}dt\leq
  \tilde C_q\,\varepsilon^{1/q}L^d.
\end{equation*}
This completes the proof of Theorem~\ref{thm:wegner}.\qed
\section{Concluding remarks} \label{sec:conclusion}

\subsection{Remarks on multi-scale analysis} \label{sec:msa}

It is well known to experts in the field that Lifshitz tails of the IDS, more precisely, a bound as in Corollary~\ref{le:3}, combined with a Wegner estimate such as Theorem~\ref{thm:wegner} lead to a proof of spectral and dynamical localization via multi-scale analysis (MSA). For the sake of reaching a broader audience we include some discussion of the strategy and additional tools which are behind this.

Two very convenient references for this discussion are \cite{Stollmann} and \cite{Klein}, which both make a point of thoroughly discussing detailed input assumptions which make MSA work and thus apply to a wide range of models, including ours. We follow the survey article \cite{Klein} here, as the results presented there are based on \cite{GeKl}, where it was shown that MSA leads to dynamical localization in the form (\ref{eq:dynloc}), the strongest result which has been obtained via MSA.

The required assumptions singled out in \cite{Klein} can be divided into deterministic and probabilistic assumptions. The deterministic assumptions listed in the following hold for classes of Schr{\"o}dinger operators much larger than what we require here. For discussion and references on their proofs see \cite{Klein} or \cite{Stollmann}. We use the same acronyms as \cite{Klein}.

(i) A property leading to the existence of suitable generalized eigenfunction expansions (SGEE).

(ii) A Simon-Lieb inequaltiy (SLI) relating resolvents at different scales.

(iii) An eigenfunction decay inequality (EDI) providing estimates for generalized eigenfunctions in terms of finite volume resolvents.

The required probabilistic properties are

(iv) $\Z^d$-ergodicity, which could be relaxed as discussed in \cite{Klein}, but is clearly satisfied for our model.

(v) Independence at distance (IAD), trivially satisfied in our model due to the non-overlap assumption for the single-site terms.

(vi) A Wegner estimate of the form (\ref{eq:wegnerest}), (\ref{eq:Wegnercl}). In \cite{Klein} the stronger form  $\E(\tr\chi_I(H_{\omega,L})) \le C_\alpha |I| L^{d}$ of the Wegner estimate (linear in the interval length) is used as an assumption, but all arguments can be modified to only require the slightly weaker (\ref{eq:wegnerest}), see Remark~4.6 in \cite{Klein}.

For random Schr{\"o}dinger operators satisfying all these assumptions it was shown in \cite{GeKl} that a certain ``suitability''-property of the finite volume resolvent, sometimes referred to as an {\it initial length estimate}, implies localization in the form claimed in Theorem~\ref{thm:main}. To establish this initial length estimate, the Lifshitz-tail bound from Corollay~\ref{le:3} is used. The argument behind this is well known, with details found, for example, in \cite{Stollmann}. Thus, we only outline the two main steps:

First, by a Neumann bracketing argument one deduces the following from (\ref{eq:Lifshitzbound}): For any $\xi>0$ and $\beta \in (0,1)$ there exists $L_2=L_2(\beta,\xi)$ such that
\[ \PP(H_{\omega,L} \:\mbox{has an eigenvalue less than}\: E_0+L^{\beta-1}) \le L^{-\xi}\]
for all $L\ge L_2$. The main difference to (\ref{eq:Lifshitzbound}) is that one trades in a larger distance of eigenvalues to $E_0$ for less, but still sufficient, probability.

Second, for $E\in I:= [E_0, E_0+L^{\beta-1}/2]$ one may now use a Combes-Thomas estimate (which holds for very general semi-bounded Schr{\"o}dinger operators, and thus certainly in our setting) to turn this into an initial length estimate such as the suitability property required in Theorem~5.4 of \cite{Klein}. This allows to start the MSA machine which leads to all the results stated in Theorem~\ref{thm:main}.

The quadruple MSA needed to prove Theorem~5.4 in \cite{Klein} is carried out in \cite{GeKl}. An equally self-contained but somewhat less refined MSA scheme is provided in \cite{Stollmann}. Here MSA is iterated twice, obtaining pure point spectrum with exponentially decaying eigenfunctions in the first run-through and using a second MSA argument to prove dynamical localization in the form
\begin{equation} \label{eq:dlstollmann}
\E \left( \sup_{|g|\le 1} \| |X|^p g(H_{\omega}) \chi_{I}(H_{\omega}) \chi_y\| \right) < \infty
\end{equation}
for all $p>0$ on an interval $I=[E_0, E_0 +\delta(p)]$ with $\delta(p)>0$ depending on $p$. This is much weaker than (\ref{eq:dynloc}), but we still consider the presentation in \cite{Stollmann} as a very accessible introduction into the mathematics of MSA for non-experts.

\subsection{Related results and problems} \label{sec:discussion}

(i) The Wegner estimate and the Lifshitz tails, and therefore our main
result Theorem~\ref{thm:main}, hold under weaker assumptions on the
distribution $\mu$ of the displacements. E.g. the proof as written in
Section~\ref{sec:wegner} only requires that $\mu$ has a $C^1$-density
$\rho$ near the corners ${\mathcal C}$. This is made possible through
the introduction of the cut-off $\eta$ supported near the corners in
(\ref{eq:cutoff}).

It is evident from (\ref{eq:4b}) that the $C^1$ condition for $\rho$
is only needed in the radial direction with respect to the corners. In
fact, similar to \cite{GhKl} we could allow distributions supported on
a suitable submanifold. Examples would be the uniform distribution
supported on a ``cross'' in $d=2$ or, in general dimension, uniform
distribution supported on the boundary $\partial G$ of
$G=(-d_{max},d_{max})^d$.

(ii) Under our assumptions and in $d\ge 2$, among all {\it periodic} configurations $\omega \in \overline{G}^{\Z^d}$, $\omega^*$ as defined in (\ref{eq:omegastar}) is, up to translations, the unique minimizer in the sense that $\inf \sigma(H_{\omega^*}) = \inf \Sigma$. This was proven in \cite{BLS2} and enters the argument in \cite{KN2} leading to Corollary~\ref{le:3} and the Lifshitz tail estimate Theorem~\ref{thr:1} above.

As also shown in \cite{BLS2}, in $d=1$ there are many periodic minimizing configurations. This has strong consequences for the IDS at the bottom of the spectrum. An extreme case is given by the 1D {\it Bernoulli displacement model}, i.e.\ $\mu = \frac{1}{2} \delta_{d_{max}} + \frac{1}{2} \delta_{-d_{max}}$, whose IDS satisfies the lower bound $N(E_0+\varepsilon) \ge C/\ln^2 \varepsilon$. This singular behavior is the extreme opposite of a Lifshitz tail.

(iii) The previous remark might mislead into expecting that the Bernoulli displacement model is not localized at low energy. Spectral localization for the 1D random displacement model (at all energies and for arbitrary non-trivial distribution of the displacements) has been proven in \cite{BS} and \cite{Sims} (using methods of \cite{DSS}). These methods are completely different from what is available for $d>1$, and, in particular, do not require smallness of the IDS. One uses dynamical systems tools such as results on products of random matrices, in particular Furstenberg's theorem. In fact, one finds that the Lyapunov exponents are positive with the possible exception of a discrete set of energies. As far as dynamical localization is concerned, it might be violated at those critical energies. For examples of this see \cite{JSS, DLS}. Away from the critical energies, however, one also has dynamical localization.

In the case of the one-dimensional Bernoulli displacement model, the energy $E_0$ provides a new example of a critical energy. This is seen as follows: By the results of \cite{BLS2}, for any $\varepsilon>0$ almost surely there is a solution $u_0$ of $H_{\omega}u=E_0 u$ and $C>0$ such that $\frac{1}{C}\exp(-x^{-1/2-\varepsilon}) \le |u_0(x)| \le C \exp(x^{1/2+\varepsilon})$. Using the lower bound and the standard reduction of order argument one shows that there is a second linearly independent solution $u_1$ which satisfies the same upper bound. This shows that the transfer matrices grow sub-exponentially. Thus, the Lyapunov exponent at $E_0$ vanishes.

(iv) An interesting open problem arises from cases of the random displacement model where, in the language used in the introduction, alternative (ii) holds, i.e.\ where $E_0(a)$ vanishes identically. Examples for this (non-generic) situation can be constructed as follows:

Let $0< \varphi \in C^{\infty}(\Lambda_1(0))$ be constant near the boundary (but not constant throughout $\Lambda_1(0)$). In the definition (\ref{eq:hamiltonian}), (\ref{eq:potential}) of the random displacement $H_{\omega}$ model pick the single-site potential as $q=\Delta \varphi/\varphi$. By construction, this leads to alternative (ii). Moreover, for {\it every} displacement configuration $\omega$ a generalized ground state of $H_{\omega}$ to $E_0=0$ is given by $\phi = \sum_n \varphi(\cdot-n-\omega_n) \chi_{\Lambda_1(n)}$ (here we think of $\varphi$ as extended by a constant onto all of $\R^d$). Note that $1/C \le \phi \le C$ for some $C>0$ uniformly in $\omega$.

This leads to van Hove behavior of the IDS at $E_0=0$, i.e.
\[ \frac{1}{C} E^{d/2} \le N(E) \le CE^{d/2} \]
for $E>0$, which follows with the same argument as provided for a closely related example in Section~3 of \cite{KN1}.

It would be interesting to know if this can generate non-trivial transport (and thus prevent dynamical localization). This is the case in dimension $d=1$. Starting with $\phi$, the reduction of order argument provides a second, linearly independent, solution which grows at most linearly. Thus, the transfer matrix also grows linearly. By Corollary~2.1 in \cite{DLS}, which is based on work in \cite{DT}, this implies that the time-averaged moments $|X|^p$ of suitable solutions of the time-dependent Schr{\"o}dinger equation are bounded below by $C T^{(p-5)/2}$. This rules out dynamical localization in the sense of (\ref{eq:dlstollmann}) for $p>5$. Clearly, there are multiple obstacles to extending these methods to higher dimension.

We finally remark that the methods of \cite{DT} and \cite{DLS} do not suffice to obtain non-trivial transport under the sub-exponential growth bounds on the transfer matrix discussed in remark (iii) above.

\vspace{.3cm}

\noindent {\bf Acknowledgements:} F.\ K., S.\ N.\ and G.\ S.\
gratefully acknowledge stays at the Banff International Reserach
Station, Centre Interfacultaire Bernoulli (supported by the SNFS) at
{\'E}cole Polytechnique F{\'e}d{\'e}rale de Lausanne, and Mathematisches
Forschungsinstitut Oberwolfach, which supported their collaboration on
this project.  G.\ S.\ also is thankful for hospitality at the
University of Tokyo. We thank the referees for useful suggestions to improve our presentation.

\end{document}